\documentclass[runningheads]{llncs}

\usepackage{cite}
\usepackage{graphicx}
\usepackage{subcaption}
\usepackage{amsmath}

\usepackage{array}
\usepackage{url}
\usepackage{epsfig}
\usepackage{amsfonts,amssymb,multirow,bigstrut,booktabs,ctable,latexsym}
\usepackage[mathscr]{eucal}
\usepackage{verbatim}

\usepackage{amsthm}
\usepackage{algpseudocode}
\usepackage{algorithmicx} 
\usepackage{algorithm}
\usepackage{enumitem}

\begin{document}
	\newcommand{\argmin}{\arg\!\min}
	\newcommand{\argmax}{\arg\!\max}
	\newcommand{\st}{\text{s.t.}}
	\newcommand*\diff{\mathop{}\!\mathrm{d}}

	\newtheorem{assumption}{Assumption}
%
\title{A Game-Theoretic Framework for Controlled Islanding in the Presence of Adversaries}

\titlerunning{Controlled Islanding in the Presence of Adversaries}
\author{Luyao Niu\inst{1}$^{\dagger}$, Dinuka Sahabandu\inst{2}$^{\dagger}$, Andrew Clark\inst{1},\and Radha Poovendran\inst{2}}
\authorrunning{Niu et al.} 
\institute{Department of Electrical and Computer Engineering, \\
Worcester Polytechnic Institute, Worcester MA 01609, USA
\and Network Security Lab, Department of Electrical and Computer Engineering,\\ University of Washington, Seattle, WA 98195-2500, USA\\
\email{\{lniu,aclark\}@wpi.edu,\{sdinuka,rp3\}@uw.edu}
\thanks{This work was supported by AFOSR grant FA9550-20-1-0074 and NSF grant CNS-1941670. $^{\dagger}$Authors contributed equally to this work.}
}

\maketitle

\begin{abstract}

Controlled islanding effectively mitigates cascading failures by partitioning the power system into a set of disjoint islands. In this paper, we study the controlled islanding problem of a power system under disturbances introduced by a malicious adversary. We formulate the interaction between the grid operator and adversary using a game-theoretic framework. The grid operator first computes a controlled islanding strategy, along with the power generation for the post-islanding system to guarantee stability. The adversary observes the strategies of the grid operator. The adversary then identifies critical substations of the power system to compromise and trips the transmission lines that are connected with compromised substations. For our game formulation, we propose a double oracle algorithm based approach that solves the best response for each player. We show that the best responses for the grid operator and adversary can be formulated as mixed integer linear programs. In addition, the best response of the adversary is equivalent to a submodular maximization problem under a cardinality constraint, which can be approximated up to a $(1-\frac{1}{e})$ optimality bound in polynomial time. We compare the proposed approach with a baseline where the grid operator computes an islanding strategy by minimizing the power flow disruption without considering the possible response from the adversary. We evaluate both approaches using IEEE 9-bus, 14-bus, 30-bus, 39-bus, 57-bus, and 118-bus power system case study data. Our proposed approach achieves better performance than the baseline in about $44\%$ of test cases, and on average it incurs about 12.27 MW less power flow disruption.

\end{abstract}



\section{Introduction}

The electric power system is a complex large-scale network that delivers electricity to customers. Modern power systems leverage Internet of Things (IoT) technologies and have integrated information and communication components \cite{abur2004power}, leading to the smart grid paradigm. However, incorporating cyber components exposes the power system to malicious cyber attacks \cite{mo2011cyber,liu2012cyber}. For example, the service outage incurred by the Ukrainian electric company in 2015 was caused by malicious cyber attacks \cite{case2016analysis}.

Cyber attacks impact power systems by biasing the decision of the power grid operator, masking physical outages, and/or causing malfunctions of system components \cite{wang2013cyber}. These disturbances can potentially lead to \emph{cascading failures}. In a cascading failure, the outage of one component, e.g., a transmission line, shifts the load to other connected components, making them overload and fail. Power systems are under increasing risks of cascading failures since they are operated close to their capacity limits so as to meet ever-increasing electricity demands. Cascading failures can cause catastrophic economic consequences; the 2003 North American blackout, for example, left more than 55 million people in dark and caused 10 billion dollars loss \cite{muir2004final}. An intelligent adversary may therefore take advantage of cascading failures to cause severe damage to power systems using limited resources.


Controlled islanding has been demonstrated to be an effective countermeasure against cascading failures \cite{you2004slow}. Controlled islanding determines a subset of transmission lines to be tripped to partition the power system into multiple subsystems, following a disturbance such as transmission line outage.

Various techniques \cite{you2004slow,sun2003splitting,ding2012two,patsakis2019strong,demetriou2018real,liu2018controlled} have been proposed for designing controlled islanding strategies with different criteria such as power flow disruption and power imbalance. To the best of our knowledge, however, there has been little study focusing on controlled islanding for power systems in the presence of malicious adversaries. Different from exogenous causes such as natural disasters and increasing load demand \cite{pidd2012india}, intelligent adversaries can infer the islanding strategy of the grid operator and deliberately trip transmission lines to make the islands ineffective or even unstable. In the 2015 Ukrainian blackout, the adversaries compromised the substations and leveraged the strategies of the grid operator against the power system \cite{case2016analysis}. 


In this paper, we propose a game-theoretic model of controlled islanding to mitigate cascading failures in the presence of a malicious adversary. The adversary can compromise a subset of substations of the power system, and trip the transmission lines that are connected to the compromised substations. The grid operator aims at preventing cascading failure triggered by the adversary by implementing a controlled islanding strategy and designing corresponding post-islanding strategies. We make the following contributions:
\begin{itemize}
    \item We model the interaction between the grid operator and adversary as a Stackelberg game, which we formulate as a mixed integer nonlinear program. 
    \item We propose a double oracle algorithm based approach to solve for the strategies of the grid operator. The proposed approach iteratively computes the best response of each player.
    \item We analyze the best response for each player, and formulate it as a mixed integer linear program. In addition, we show the equivalence between the adversary's best response and a submodular function maximization problem with a cardinality constraint. A greedy algorithm can then be used to approximately compute the adversary's best response in polynomial time.
    \item We evaluate our proposed approach using IEEE 9-bus, 14-bus, 30-bus, 39-bus, 57-bus, and 118-bus power system case study data. We show that on average the power system performs better in $44\%$ of the test cases and incurs $12.27$~MW less power flow disruption when using the proposed approach, compared with a baseline that ignores the presence of an adversary.
\end{itemize}

The remainder of this paper is organized as follows. Related literature is reviewed in Section \ref{sec:related}. Preliminary background on power system model and Stackelberg games is presented in Section \ref{sec:prelim}. Section \ref{sec:formulation} gives the models for the adversary and grid operator, and maps the problem to a Stackelberg game. We present the solution approach in Section \ref{sec:sol}. Numerical evaluation results are presented in Section \ref{sec:simulation}. We conclude the paper in Section \ref{sec:conclusion}.
\section{Related Work}\label{sec:related}


Computing controlled islanding strategies for power systems under large disturbances has been extensively studied. Typical approaches include slow coherency based islanding \cite{you2004slow}, ordered binary decision diagram (OBDD) methods \cite{sun2003splitting}, two-step spectral clustering technique \cite{ding2012two}, weak submodularity based controlled islanding \cite{liu2018controlled}, and mixed integer program based approaches \cite{patsakis2019strong,demetriou2018real}. These works study the controlled islanding by assuming the disturbance has been detected and fixed. When there exists an adversary who can intelligently adjust its strategy, the islanding strategies computed using the aforementioned contributions need to be adjusted also to incorporate the possible response from the adversary.

Malicious attacks targeting power system have been reported and studied. The malicious attacks can be roughly classified into two categories. The first category of attacks manipulates the grid topology via transmission line switching \cite{nedic2006criticality,delgadillo2009analysis} and compromising substations \cite{zhu2014revealing}. Another category of attacks targets at the cyber components such as false data injection attack \cite{yang2013false,deng2015defending}. In this paper, we consider a malicious adversary that compromises substations using cyber attacks and trips transmission lines. Different from existing literature, the adversary model studied in this paper not only identifies the critical components in the power system, but also considers the possible corrective action taken by the grid operator. In addition, existing topological models may discard the power system dynamics to simplify the computations \cite{hasan2018vulnerability}, while this paper takes the physical properties of power systems into consideration.

In this paper, we map the interaction between the grid operator and adversary to a Stackelberg game. Stackelberg games have been widely used to model real-world security applications such as airport protection \cite{pita2009using}. To compute the Stackelberg equilibrium \cite{conitzer2016stackelberg} of the game in this paper, we propose a double oracle algorithm based approach \cite{mcmahan2003planning}. Double oracle algorithm has been widely used to solve games of large-scale \cite{jain2011double,karwowski2020double,bosansky2014exact}, due to the advantage that it avoids the enumeration over all possible strategies for the players.
\section{Model and Preliminaries}\label{sec:prelim}


\subsection{Power System Model}

A power system with $B$ substations and $L$ transmission lines can be described by a graph $\mathcal{G}=(\mathcal{B},\mathcal{L})$, where $\mathcal{B}=\{1,\ldots,B\}$ is the set of substations and $\mathcal{L}\subseteq\mathcal{B}\times\mathcal{B}$ is the set of transmission lines. A transmission line $l=(i,j)\in\mathcal{L}$ if substations $i$ and $j$ are connected via $l$. We define the set of neighboring substations for each $i\in\mathcal{B}$ as $\mathcal{T}(i)=\{j:(i,j)\in\mathcal{L}\}$. The power injected to substation $i$ is denoted as $g_i$, and the power drawn from substation $i$ is denoted as $d_i$.

We consider DC power flow in the power system. The power flow $P_{i,j}$ through each transmission line $(i,j)$ is calculated as
\begin{equation}\label{eq:DC power flow}
    P_{i,j}=S_{i,j}(\theta_i-\theta_j),~\forall (i,j)\in\mathcal{L}
\end{equation}
where $S_{i,j}$ is the electrical susceptance of transmission line $(i,j)$, and $\theta_i$, $\theta_j$ are the voltage angles at substations $i$ and $j$, respectively. Each substation $i\in\mathcal{B}$ respects the power flow conservation law given as
\begin{equation}\label{eq:power balance}
    \sum_{j\in\mathcal{T}(i)}P_{j,i} + g_i - d_i=0,~\forall i\in\mathcal{B}.
\end{equation}

Power generators exhibit varying behaviors following a large disturbance. Two generators are said to be \emph{coherent} if their rotor angle deviations are within a certain tolerance \cite{haque1990identification}. To maintain the internal stability of the power system, the coherent generators need to be connected, while the non-coherent ones must be separated during islanding. In this paper, we assume that the set of power generators are classified into $K$ coherent groups. Detailed techniques on computing coherent groups can be found in \cite{chow1982time}.

There are various metrics that have been proposed to measure the performance of power system when incurring disturbance. Typical metrics include power flow disruption \cite{peiravi2009fast,yang2007novel} and power imbalance \cite{zhao2003study,sun2003splitting}. Minimum power flow disruption improves the transient stability of the system and reduces the risk of overloading transmission lines \cite{henner1980network}. In this paper, we adopt power flow disruption as the performance metric, which is defined as
\begin{equation}\label{eq:power flow disruption}
    R(\mathcal{S})=\sum_{(i,j)\in\mathcal{L}}(1-z_{i,j})\frac{|P_{i,j}|+|P_{j,i}|}{2},
\end{equation}
where $\mathcal{S}\subseteq\mathcal{L}$ represents the set of tripped transmission lines. Parameter $z_{i,j}=1$ if $(i,j)\in\mathcal{L}\setminus\mathcal{S}$ and $z_{i,j}=0$ if $(i,j)\in\mathcal{S}$.

\subsection{Stackelberg Game}\label{sec:Stackelberg}

Game theory models the interaction among multiple players. Consider a game consisting of two players, denoted as Player 1 and Player 2. Players 1 and 2 have their action spaces $\mathcal{A}_1$ and $\mathcal{A}_2$, respectively. Each action of $\mathcal{A}_1$ and $\mathcal{A}_2$ is also known as the pure strategy for Player 1 and Player 2, respectively. A mixed strategy is a probability distribution over the action space. When Players 1 and 2 take strategies $s_1$ and $s_2$, respectively, they obtain utilities $U_1(s_1,s_2)$ and $U_2(s_1,s_2)$ during the interaction. 

Two-player Stackelberg games model interactions with information asymmetry, where Player 1 moves first by committing to a strategy, and Player 2 observes the strategy committed by Player 1 and chooses its strategy to maximize $U_2(\cdot,\cdot)$. Player 1 and Player 2 are also known as the leader and follower, respectively. 

The solution concept of Stackelberg game is called Stackelberg equilibrium. We say strategies $s_1^*$ and $s_2^*$ for Players 1 and 2 comprise a Stackelberg equilibrium if $s_1^*=\argmax_{s_1} U(s_1,s_2^*)$, where $s_2^*\in \mathcal{BR}(s_1^*)$ and $\mathcal{BR}(s_1^*)=\argmax_{s_2}\{\allowbreak U_2( s_1^*,\allowbreak s_2)\}$ is the best response taken by Player 2 to $s_1^*$.
\section{Problem Formulation}\label{sec:formulation}

\subsection{Adversary Model}

We consider a power system $\mathcal{G}=(\mathcal{B},\mathcal{L})$. A malicious adversary aims at destablizing the power system and maximizing the power flow disruption. To achieve this goal, the adversary has two capabilities: (i) the adversary can compromise at most $C$ substations $\hat{\mathcal{B}}\subset\mathcal{B}$, and (ii) the adversary can trip the set of transmission lines that are connected with the compromised substations. These capabilities have been demonstrated by real-world adversaries. For instance, the adversary that initiated the attack against the Ukrainian electric system compromised the substations and thus gained control over field devices \cite{case2016analysis}. 

We assume the adversary has access to the following information. The adversary knows the grid topology $\mathcal{G}=(\mathcal{B},\mathcal{L})$ and the power flow before it trips any transmission line. In addition, the adversary can observe the strategies taken by the grid operator (we will detail the model of grid operator in Section \ref{sec:operator model}). We denote the information available to the adversary as $\mathcal{I}^a$. It has been reported that the adversary can harvest such information via cyber attacks \cite{case2016analysis}.

In the following, we define the strategy for the adversary. A pure strategy for the adversary $\tau:\mathcal{I}^a\rightarrow 2^\mathcal{B}\times2^\mathcal{L}$ maps from the information set of the adversary to a pair of compromised substations $\hat{\mathcal{B}}$ and tripped transmission lines $\hat{\mathcal{L}}$. A mixed strategy for the adversary $\tau:\mathcal{I}^a\rightarrow \Delta(2^\mathcal{B}\times2^\mathcal{L})$ maps from the information set of the adversary to a pair of probability distributions over $2^\mathcal{B}\times2^\mathcal{L}$, where $\Delta(\cdot)$ represents a probability distribution over some set. We define the set of proper adversary strategies as follows.
\begin{definition}\label{def:proper adv strategy}
We say strategy $\tau$ is proper if the following conditions hold: (i) $|\hat{\mathcal{B}}|\leq C$, and (ii) $\hat{\mathcal{L}}\subseteq\{(i,j)\in\mathcal{L}:i\in\hat{\mathcal{B}}\text{ or }j\in\hat{\mathcal{B}}\}$.
\end{definition}

The adversary computes its strategy $\tau$ as
\begin{subequations}\label{eq:adv}
\begin{align}
    \max_{\tau}~&R(\mathcal{S})\\
    \text{subject to }&\tau\text{ is proper}
\end{align}
\end{subequations}
where $R(\mathcal{S})$ is defined in Eqn. \eqref{eq:power flow disruption}, $\mathcal{S}=\hat{\mathcal{L}}\cup\Tilde{\mathcal{L}}\subseteq\mathcal{L}$, and $\Tilde{\mathcal{L}}$ is the set of transmission lines tripped by the grid operator (we will introduce $\Tilde{\mathcal{L}}$ later in Section \ref{sec:operator model}).

\subsection{Grid Operator Model}\label{sec:operator model}

In this subsection, we present the model of the grid operator. The goal of the grid operator is to protect the power system $\mathcal{G}=(\mathcal{B},\mathcal{L})$ when large disturbance is incurred. The grid operator has the following control capabilities. The grid operator can trip a subset of transmission lines $\Tilde{\mathcal{L}}\subset\mathcal{L}\setminus\hat{\mathcal{L}}$ to partition the power system into a collection of subsystems $\{\mathcal{G}_k\}_{k=1}^K$, where $\mathcal{G}_k=(\mathcal{B}_k,\mathcal{L}_k)$, $\mathcal{B}_k\subset\mathcal{B}$, and $\mathcal{L}_k\subset\mathcal{L}$. A subsystem $\mathcal{G}_k$ is also known as an island. After the power system is partitioned into subsystems, the grid operator controls the power injection $g_i$ from each generator at each generation substation $i\in\mathcal{B}$. 

We assume that the grid operator knows the grid topology and has perfect observation over the power system so that it can monitor the parameters such as the voltage angle at each substation, the power flow at each transmission line, the power injection from the generators, and the power drawn by the load demands. Additionally, the grid operator can compute the set of generator coherent groups. We denote the information available to the grid operator as $\mathcal{I}^o$.

A pure strategy for the grid operator is defined as $\mu:\mathcal{I}^o\rightarrow 2^\mathcal{L}\times \mathbb{R}^N$ that maps from $\mathcal{I}^o$ to the set of possibly tripped transmission lines and the space of power generations. Note that here the set of open transmission lines are tripped by the grid operator, and is different from those tripped by the adversary. A mixed strategy for the grid operator is defined as $\mu:\mathcal{I}^o\rightarrow \Delta(2^\mathcal{L}\times \mathbb{R}^N)$. We define the set of proper strategies for the grid operator.
\begin{definition}\label{def:proper operator strategy}
A strategy $\mu$ for the grid operator is proper if the following conditions hold: (i) strategy $\mu$ partitions the power system into disjoint subsystems, i.e., $\mathcal{B}_k\cap\mathcal{B}_{k'}=\emptyset$ and $\mathcal{L}_k\cap\mathcal{L}_{k'}=\emptyset$, (ii) the generators belonging to the same coherent group are within the same subsystem, (iii) each subsystem $\mathcal{G}_k$ is connected, (iv) the post-islanding power generation and voltage angle are within the generation capacity for each generator and voltage angle bound for each substation, respectively, (v) the post-islanding power flow does not exceed the transmission capacities for all transmission lines, and (vi) the post-islanding power generation meets the load demand.
\end{definition}


\subsection{Interaction Model Between the Grid Operator and Adversary}

In this subsection, we present the interaction model between the grid operator and adversary. We denote the mixed strategy of the grid operator as $\mu:\mathcal{I}^o\rightarrow \Delta(2^\mathcal{L},\mathbb{R}^N)$. The adversary observes strategy $\mu$ of the grid operator by intruding into the power network and learning the strategies of the grid operator, and then computes a proper strategy $\tau$. Then the adversary executes its attack strategy $\tau$ so as to destabilize the power system and maximize the power flow disruption. Once the grid operator detects the disturbance caused by the adversary, it samples a pair $(\Tilde{\mathcal{L}},g)$ following strategy $\mu$, and implements the sampled action to partition the system into a collection of subsystems.

There exists information asymmetry during the interaction between the grid operator and adversary. The adversary observes the strategy of the grid operator, while the grid operator has no information on the strategy of the adversary. This information asymmetry is captured by the Stackelberg game as described in Section \ref{sec:Stackelberg}. During this interaction, since the grid operator computes strategy $\mu$ first, it becomes the leader in the Stackelberg game. The adversary, who observes the leader's strategy, is the follower in this setting. 

The problem investigated in this paper is stated as follows.
\begin{problem}\label{prob:formulation}
Consider a power system $\mathcal{G}=(\mathcal{B},\mathcal{L})$. Synthesize a proper strategy $\mu$ for the grid operator that minimizes the power flow disruption, given that the adversary observes $\mu$ and computes its best response to $\mu$, i.e.,
\begin{subequations}\label{eq:Stackelberg}
\begin{align}
    \min_{\mu}~&\mathbb{E}_{\mu}[R(\mathcal{S}(\mu,\tau))]\\
    \text{subject to }&\mu\text{ is proper}\\
    &\tau\in\mathcal{BR}(\mu)
\end{align}
\end{subequations}
where $\mathbb{E}_{\mu}[\cdot]$ denotes the expectation with respect to $\mu$ and $\mathcal{S}(\mu,\tau)=\hat{\mathcal{L}}\cup\Tilde{\mathcal{L}}\subseteq\mathcal{L}$ is the set of tripped transmission lines that is jointly determined by $\mu$ and $\tau$.
\end{problem}
Note that the interaction between the grid operator and adversary is zero-sum. We can thus establish the existence of Stackelberg equilibrium strategies $\mu$ and $\tau$ of the game in Eqn. \eqref{eq:Stackelberg} using \cite[Section 2]{conitzer2016stackelberg}.

\section{Solution Approach}\label{sec:sol}

In this section, we present the solution approach to Problem \ref{prob:formulation}. We prove that the sets of proper strategies $\mu$ and $\tau$ can be mapped to sets of mixed integer constraints. Then the optimization problem in Eqn. \eqref{eq:Stackelberg} is formulated as a mixed integer nonlinear program. We propose a double oracle algorithm based approach to compute the Stackelberg equilibrium strategies. The proposed approach computes the best response of each player in each iteration. We show that the best responses for both players can be formulated as mixed integer linear programs.

\subsection{Mixed Integer Nonlinear Bi-level Optimization Formulation}

In this subsection, we first map the set of proper strategies $\mu$ and $\tau$ to a set of mixed integer constraints. We then rewrite Eqn. \eqref{eq:Stackelberg} as a mixed integer nonlinear bi-level optimization problem.

We let $y_{i}$ be a binary variable representing if substation $i\in\mathcal{B}$ is compromised by the adversary ($y_i=1$) or not ($y_i=0$). We then have the following constraints:
\begin{equation}\label{eq:constr attack substation}
    \sum_{i\in\mathcal{B}}y_i\leq C,\quad y_i\in\{0,1\},~\forall i\in\mathcal{B}.
\end{equation}

We define $z_{i,j}^a$ as an indicator function for each transmission line $(i,j)$ to represent if transmission line $(i,j)$ is tripped ($z_{i,j}^a=0$) or not ($z_{i,j}^a=1$) by the adversary. Note that the adversary can trip a transmission line $(i,j)$ if and only if substation $i$ or $j$ is compromised. We formulate this property as 
\begin{subequations}\label{eq:constr attack indicator}
\begin{align}
    &z_{i,j}^a\in\{0,1\},~z_{i,j}^a=z_{j,i}^a,~\forall (i,j)\in\mathcal{L}\label{eq:constr adv line indicator 1}\\
    &z_{i,j}^a+y_i+y_j\geq 1,~\forall i\in\mathcal{B},~\forall (i,j)\in\mathcal{L}.\label{eq:constr adv line indicator 2}
\end{align}
\end{subequations}

We denote $y$ as the vector obtained by stacking $y_i$ for all $i\in\mathcal{B}$ and $z^a$ as the vector obtained by stacking $z_{i,j}^a$ for all $(i,j)\in\mathcal{L}$. We characterize the relations given by Eqn. \eqref{eq:constr attack substation} and \eqref{eq:constr attack indicator} as follows.
\begin{lemma}\label{lemma:adv proper}
The set of proper strategies for the adversary is equal to the set of feasible solutions $(y,z^a)$ to Eqn. \eqref{eq:constr attack substation} and \eqref{eq:constr attack indicator}.
\end{lemma}
\begin{proof}
We prove the statement using Definition \ref{def:proper adv strategy}. Consider condition (i) in Definition \ref{def:proper adv strategy}. Since $y_i\in\{0,1\}$ for all $i\in\mathcal{B}$, we have that if $y$ is feasible to Eqn. \eqref{eq:constr attack substation}, then at most $C$ substations can be compromised. Consider condition (ii) in Definition \ref{def:proper adv strategy}. By Definition \ref{def:proper adv strategy} and the definitions of $z^a$ and $y$, we have that $z_{i,j}^a=0$ only if $y_i+y_j\geq 1$. However, $y_i+y_j\geq 1$ does not necessarily imply that $z_{i,j}^a=0$, i.e., the adversary can choose to not trip transmission line $(i,j)$ even if substation $i$ or $j$ is compromised. In addition, $z_{i,j}^a=1$ must hold if $y_i+y_j=0$. Summarizing these three possible scenarios, we have that $z_{i,j}^a$, $y_i$, and $y_j$ cannot be zero simultaneously, which is equivalent to Eqn. \eqref{eq:constr adv line indicator 2}. Combining the arguments above yields the lemma.
\end{proof}

Consider the set of proper strategies for the grid operator. Note that the pure strategy space for the grid operator grows exponentially with respect to the number of transmission lines $L$. To this end, we define a set of variables for each transmission line as a compact representation of the set of proper strategies. 

Let $x_{i,k}$ be a binary indicator representing if substation $i$ is included in subsystem $k$ ($x_{i,k}=1$) or not ($x_{i,k}=0$) for all $i\in\mathcal{B}$ and $k=1,\ldots,K$. In addition, we let $w_{i,j,k}$ be an indicator, representing if transmission line $(i,j)\in\mathcal{L}$ is included ($w_{i,j,k}=1$) in subsystem $\mathcal{G}_k=(\mathcal{B}_k,\mathcal{L}_k)$ or not ($w_{i,j,k}=0$). For each transmission line $(i,j)$, we define $z^o_{i,j}$ as an indicator representing if transmission line $(i,j)\in\mathcal{L}$ is tripped ($z^o_{i,j}=0$) or not ($z^o_{i,j}=1$) by the grid operator. We then formulate the constraints as
\begin{subequations}\label{eq:constr indicator}
\begin{align}
    &w_{i,j,k}\in\{0,1\},~w_{i,j,k}\leq x_{i,k},~w_{i,j,k}\leq x_{j,k},~\forall (i,j)\in\mathcal{L},~\forall k=1,\ldots,K\label{eq:constr indicator 1}\\
    &z_{i,j}^o=\sum_{k=1}^Kw_{i,j,k},~z_{i,j}^o\in\{0,1\},~z_{i,j}^o=z_{j,i}^o,~\forall (i,j)\in\mathcal{L}\label{eq:constr indicator 2}\\
    &\sum_{k=1}^Kx_{i,k}\leq 1,~\forall i\in\mathcal{B}\label{eq:constr indicator 3}\\
    &x_{i,k}\in\{0,1\},~\forall i\in\mathcal{B},~\forall k=1,\ldots,K\label{eq:constr indicator 5}\\
    &z_{i,j}^o\leq z_{i,j}^a,~\forall (i,j)\in\mathcal{L}\label{eq:constr indicator 6}
\end{align}
\end{subequations}
Eqn. \eqref{eq:constr indicator 6} captures the fact that the grid operator takes islanding action after the adversary executes the malicious attack. Hence, the grid operator cannot open a transmission line that has been tripped by the adversary.

Given the generator coherent groups, we let indicator $v_{i,k}=1$ if generation substation $i$ is set as the reference generator and belongs to subsystem $\mathcal{G}_k$. Then using the coherent group, we can let
\begin{equation}\label{eq:constr coherent}
    x_{j,k}=v_{i,k},~\forall i,j\in\mathcal{C}_k,
\end{equation}
where $\mathcal{C}_k$ represents the $k$-th generator coherent group. In addition, each subsystem $\mathcal{G}_k$ is required to be connected. In order to incorporate this constraint, we define an auxiliary flow $f_{i,j,k}$ on each transmission line $(i,j)$ of subsystem $k$. Then the auxiliary flow should respect the flow conservation law given as
\begin{subequations}\label{eq:constr connectivity}
\begin{align}
    &0\leq f_{i,j,k}\leq Zz^o_{i,j},~\forall (i,j)\in\mathcal{L}\label{eq:constr connectivity 1}\\
    &v_{i,k}\sum_{j\in\mathcal{B}}x_{j,k} - x_{i,k} + \sum_{j\in\mathcal{T}(i)}f_{j,i,k}=\sum_{j\in\mathcal{T}(i)}f_{i,j,k},~\forall i\in\mathcal{B},k=1,\ldots,K\label{eq:constr connectivity 2}
\end{align}
\end{subequations}
where $Z$ is a sufficiently large positive constant. The first term of Eqn. \eqref{eq:constr connectivity 2} implies that $\sum_{j\in\mathcal{B}}x_{j,k}$ amount of auxiliary flow originates from the reference generator of subsystem $\mathcal{G}_k$. The second term of Eqn. \eqref{eq:constr connectivity 2} indicates that one unit of auxiliary flow is consumed at substation $i$. The remaining two terms of Eqn. \eqref{eq:constr connectivity 2} capture the incoming and outgoing auxiliary flows at substation $i$.

Relations given in Eqn. \eqref{eq:constr indicator} to Eqn. \eqref{eq:constr connectivity} characterize the topological properties of each subsystem $\mathcal{G}_k$. In the following, we characterize the physical properties including the power flow and voltage angle in the power system after controlled islanding is implemented.

Each generator is constrained by its generation capacity modeled as
\begin{equation}\label{eq:constr generation}
    \underline{g}_{i}\leq g_i\leq \bar{g}_{i},~\forall i\in\mathcal{B}
\end{equation}
where $\underline{g}_{i}$ and $\bar{g}_{i}$ are the minimum and maximum power generation capacities for generation substation $i$, respectively.
We denote the post-islanding power flow on transmission line as $\Tilde{P}_{i,j}$ and voltage angle of substation $i$ as $\theta_i$. By Eqn. \eqref{eq:DC power flow}, we have that $S_{i,j}(\theta_i-\theta_j)-\Tilde{P}_{i,j}=0$ holds for all $(i,j)\in\mathcal{L}$.
To incorporate the fact that the transmission line $(i,j)$ can be tripped by the grid operator and adversary, we have that
\begin{equation}\label{eq:constr power flow}
    -(1-z_{i,j}^oz_{i,j}^a)Z\leq S_{i,j}(\theta_i-\theta_j)-\Tilde{P}_{i,j}
    \leq(1-z_{i,j}^oz_{i,j}^a)Z,
\end{equation}
where $Z$ is a sufficiently large positive constant. Taking the transmission line capacity and voltage angle bound into consideration, we have
\begin{equation}\label{eq:constr bound}
    \underline{P}_{i,j}z_{i,j}^oz_{i,j}^a\leq \tilde{P}_{i,j}\leq \bar{P}_{i,j}z_{i,j}^oz_{i,j}^a,~\forall (i,j)\in\mathcal{L},~
    \underline{\theta}_i\leq\theta_i\leq\bar{\theta}_i,~\forall i\in\mathcal{B},
\end{equation}
where $\bar{P}_{i,j}$ and $\underline{P}_{i,j}$ are the maximal and minimal power flow capacity for transmission line $(i,j)$, and $\underline{\theta}_i$ and $\bar{\theta}_i$ are respectively the minimum and maximum voltage angle at substation $i$. Using Eqn. \eqref{eq:constr bound}, we observe that the only feasible power flow through a tripped transmission line is zero. By Eqn. \eqref{eq:power balance}, the power balance at each substation $i$ is modeled as
\begin{equation}\label{eq:constr power balance}
    \sum_{j\in\mathcal{T}(i)}\Tilde{P}_{j,i}+g_{i} - d_i=0,~\forall i\in\mathcal{B}.
\end{equation}

We denote $w$, $x$, $v$, $z^o$, $f$, $\Tilde{P}$, $g$, and $\theta$ as the vectors or matrices that are obtained by stacking $w_{i,j,k}$, $x_{i,k}$, $v_{i,k}$, $z_{i,j}^o$ $f_{i,j,k}$, $\Tilde{P}_{i,j}$, $g_{i}$, and $\theta_{i}$, respectively. We characterize Eqn. \eqref{eq:constr indicator} to \eqref{eq:constr power balance} as follows.
\begin{lemma}\label{lemma:operator proper}
If variables $w$, $x$, $t$, $z^o$, $f$, $\Tilde{P}$, $g$, and $\theta$ are feasible to Eqn. \eqref{eq:constr indicator} to Eqn. \eqref{eq:constr power balance}, then these variables represent a proper strategy for the grid operator as given in Definition \ref{def:proper operator strategy}
\end{lemma}
\begin{proof}
Consider variables $w$, $x$, $t$, $z^o$, $f$, $\Tilde{P}$, $g$, and $\theta$ that are feasible to Eqn. \eqref{eq:constr indicator} to Eqn. \eqref{eq:constr power balance}. We then verify that conditions (i)-(vi) in Definition \ref{def:proper operator strategy} are satisfied.

\underline{\emph{Satisfaction of Condition (i).}} Suppose that Eqn. \eqref{eq:constr indicator} is satisfied while the subsystems $\mathcal{G}_k$ are not disjoint. Thus we have that there exists $k\neq k'$ such that $x_{i,k}=x_{i,k'}=1$ holds for some $i\in\mathcal{B}$ or $w_{i,j,k}=w_{i,j,k'}=1$ holds for some $(i,j)\in\mathcal{L}$. If $x_{i,k}=x_{i,k'}=1$ holds for some $i\in\mathcal{B}$, then Eqn. \eqref{eq:constr indicator 3} is violated. $w_{i,j,k}=w_{i,j,k'}=1$ holds for some $(i,j)\in\mathcal{L}$, then Eqn. \eqref{eq:constr indicator 2} implies that $z_{i,j}^o>1$, which leads to contradiction. Thus, condition (i) of Definition \ref{def:proper operator strategy} is satisfied.

\underline{\emph{Satisfaction of condition (ii)}.} Condition (ii) holds immediately by the definition of $x_{i,k}$, $v_{i,k}$, and Eqn. \eqref{eq:constr coherent}.

\underline{\emph{Satisfaction of condition (iii)}.} Suppose $f$ satisfies Eqn. \eqref{eq:constr connectivity} while there exists some subsystem $\mathcal{G}_k$ that is not connected. Without loss of generality, we assume that substation $i$ belonging to subsystem $\mathcal{G}_k$ is not connected with substations $j\in\mathcal{B}_k\setminus\{i\}$. If substation $i$ is not the $k$-th reference generation substation, then Eqn. \eqref{eq:constr connectivity 2} becomes $-x_{i,k}=0$, which contradicts our hypothesis that $x_{i,k}=1$. If substation $i$ is the $k$-th reference generation substation, then $v_{i,k}=1$ and Eqn. \eqref{eq:constr connectivity 2} is rewritten as $\sum_{j\in\mathcal{B}}x_{j,k} - x_{i,k}=0$, which leads contradiction since there exists $j\in\mathcal{B}_k\setminus\{i\}$ such that $x_{j,k}=1$. Therefore, we can conclude that  condition (iii) of Definition \ref{def:proper operator strategy} is satisfied when Eqn. \eqref{eq:constr connectivity} is satisfied.

\underline{\emph{Satisfaction of condition (iv)}.} Condition (iv) follows from Eqn. \eqref{eq:constr generation} and \eqref{eq:constr bound}.

\underline{\emph{Satisfaction of condition (v)}.} Consider Eqn. \eqref{eq:constr bound} for a transmission line $(i,j)$. If transmission line $(i,j)$ is tripped by either the adversary or the grid operator, then Eqn. \eqref{eq:constr bound} implies that $\tilde{P}_{i,j}=0$, which satisfies the power flow equation. If transmission line $(i,j)$ is tripped by neither the adversary nor the grid operator, then power flow $\tilde{P}_{i,j}$ satisfies Eqn. \eqref{eq:DC power flow}. The transmission line capacity constraint then immediately follows from Eqn. \eqref{eq:constr bound}.

\underline{\emph{Satisfaction of condition (vi)}.} Condition (vi) of Definition \ref{def:proper operator strategy} holds by the definitions of $\tilde{P}_{i,j}$, $g_i$, $d_i$, and Eqn. \eqref{eq:power balance}.
\end{proof}

Lemma \ref{lemma:adv proper} and \ref{lemma:operator proper} imply that we can represent the pure strategy space using a collection of variables, whose size is polynomial in terms of $B$ and $L$. Using these variables, we can rewrite optimization problem \eqref{eq:Stackelberg} as
\begin{subequations}\label{eq:bilevel}
\begin{align}
    \min_{w,x,z^o,f,\tilde{P},g,\theta}~&\mathbb{E}_\mu\left[\sum_{(i,j)\in\mathcal{L}}(1-z_{i,j}^oz_{i,j}^a)\frac{|P_{i,j}|+|P_{j,i}|}{2}\right]\label{eq:bilevel 1}\\
    \text{subject to }~& \text{Eqn. \eqref{eq:constr indicator} to Eqn. \eqref{eq:constr power balance}}\label{eq:bilevel 2}\\
    &(y,z^a)\in\argmax\sum_{(i,j)\in\mathcal{L}}(1-z_{i,j}^oz_{i,j}^a)\frac{|P_{i,j}|+|P_{j,i}|}{2}\label{eq:bilevel 3}\\
    &\text{subject to}~\text{Eqn. \eqref{eq:constr attack substation} to Eqn. \eqref{eq:constr attack indicator}}\label{eq:bilevel 4}
\end{align}
\end{subequations}
Eqn. \eqref{eq:bilevel 1} to \eqref{eq:bilevel 2} and Eqn. \eqref{eq:bilevel 3} to \eqref{eq:bilevel 4} are known as the upper and lower level of bi-level optimization program \eqref{eq:bilevel}, respectively. We remark that although $\tilde{P}$ and $\theta$ are set as decision variables in optimization program \eqref{eq:bilevel}, they are inherently determined once the grid topology and power generation $g$ are given. Therefore, the upper level of Eqn. \eqref{eq:bilevel} is interpreted as computing the partitions of the power system using $z^o$ as a corrective measure against the malicious attack. For the power system partition $z^o$, the grid operator needs to compute power generation $g$ so that there exists some feasible post-islanding DC power flow $\tilde{P}_{i,j}$ satisfies conditions (iv)-(vi) in Definition \ref{def:proper operator strategy}.

\subsection{Double Oracle Algorithm Based Approach}

In this subsection, we present a double oracle algorithm based approach to solve Problem \ref{prob:formulation}. The proposed approach alternatively solves the upper and lower level of optimization problem \eqref{eq:bilevel}, and converges to the Stackelberg equilibrium.

\begin{center}
  	\begin{algorithm}[!htp]
  		\caption{Double Oracle Algorithm for Controlled Islanding}
  		\label{algo:DO}
  		\begin{algorithmic}[1]
  			\State Initialize a set of actions  $(\mathcal{Z}^o,G)$ for the grid operator, with each $(z^o,g)\in(\mathcal{Z}^o,G)$ being feasible to Eqn. \eqref{eq:constr indicator} to Eqn. \eqref{eq:constr power balance}
  			\State Initialize a set of actions $(Y,\mathcal{Z}^a)$ for the adversary, with each $(y,z^a)\in(\mathcal{Y},\mathcal{Z}^a)$ being feasible to Eqn. \eqref{eq:constr attack substation} to Eqn. \eqref{eq:constr attack indicator}
  			\While {not converge}
  			\State Solve for $(\mu,\tau)$ by constraining the grid operator and adversary to take actions from $(\mathcal{Z}^o,G)$ and $(\mathcal{Y},\mathcal{Z}^a)$, respectively
  			\State Compute $(z^o,g)$, assuming the adversary takes strategy $\tau$
  			\State $(\mathcal{Z}^o,G)\leftarrow(\mathcal{Z}^o,G)\cup (z^o,g)$
  			\State Solve for $(y,z^a)$, assuming the grid operator takes strategy $\mu$
  			\State $(\mathcal{Y},\mathcal{Z}^a)\leftarrow(\mathcal{Y},\mathcal{Z}^a)\cup (y,z^a)$
  			\EndWhile
  		\State \Return $(\mu,\tau)$
  		\end{algorithmic}
  	\end{algorithm}
  \end{center}

Algorithm \ref{algo:DO} presents the double oracle approach. It consists of four steps. The first step is presented in lines 1 to 2 of Algorithm \ref{algo:DO}. In this step, the algorithm initializes a set of pure strategies for the grid operator and adversary, respectively. The initialized pure strategies are proper. The second step corresponds to line 4 of Algorithm \ref{algo:DO}. In this step, the algorithm solves a mixed strategy $\mu$ for the grid operator and a pure strategy $\tau$ for the adversary. The reason that pure strategy is considered for the adversary is that it is the follower in the game, whose pure strategies suffice for best response calculation \cite{conitzer2016stackelberg}. Note that here mixed strategy $\mu$ defines a probability distribution over $(\mathcal{Z}^o,G)$, rather than the full strategy space. Similarly, best response $\tau$ gives an action selected from $(\mathcal{Y},\mathcal{Z}^a)$. Lines 5 to 6 correspond to the third step of Algorithm \ref{algo:DO}. This step computes a pure strategy for the grid operator over all the feasible strategies, given that the adversary plays strategy $\tau$. The fourth step is presented in lines 7-8 of Algorithm \ref{algo:DO}, where the adversary computes its best response to mixed strategy $\mu$. The second to the last step of Algorithm \ref{algo:DO} are executed in an iterative manner. The iteration terminates when no pure strategies for the grid operator and adversary are included in line 6 and line 8. The worst-case number of iterations Algorithm \ref{algo:DO} can take to converge is $(2^{L}-1)$, which is identical to solving for the Stackelberg equilibrium using linear program \cite{conitzer2016stackelberg}. However, implementing the linear program requires constructing the action spaces of dimensions $2^L$ for the grid operator and adversary and the corresponding constraints.

Given the current set of pure strategies $(\mathcal{Z}^o,G)$ and $(\mathcal{Y},\mathcal{Z}^a)$ for the gird operator and adversary, respectively, line 4 of Algorithm \ref{algo:DO} can be formulated as
\begin{subequations}\label{eq:MINLP}
\begin{align}
    \underset{\mu,\tau,r}{\min}&\sum_{z^o\in\mathcal{Z}^o}\sum_{z^a\in\mathcal{Z}^a}\mu(z^o)\tau(z^a)\sum_{(i,j)\in\mathcal{L}}\frac{1-z^o_{i,j}z^a_{i,j}}{2}(|P_{i,j}|+|P_{j,i}|)\label{eq:MINLP obj}\\
    \text{subject to}\quad& \sum_{z^o\in\mathcal{Z}^o}\mu(z^o)=1\label{eq:MINLP constraint 1}\\
    &\mu(z^o)\in[0,1],~\forall z^o\in \mathcal{Z}^o\label{eq:MINLP constraint 2}\\
    &\sum_{z^a\in\mathcal{Z}^a}\tau(z^a)=1\label{eq:MINLP constraint 3}\\
    &\tau(z^a)\in\{0,1\},~\forall z^a\in \mathcal{Z}^a\label{eq:MINLP constraint 4}\\
    &0\leq r-\sum_{z^o\in\mathcal{Z}^o}\mu(z^o)\sum_{(i,j)\in\mathcal{L}}\frac{1-z^o_{i,j}z^a_{i,j}}{2}(|P_{i,j}|+|P_{j,i}|)\nonumber\\
    &\quad\quad\quad\quad\quad\quad\quad\leq (1-\tau(z^a))Z,~\forall z^a\in\mathcal{Z}^a\label{eq:MINLP constraint 5}\\
    &r\geq 0\label{eq:MINLP constraint 6}\\
    &\text{Eqn. \eqref{eq:constr attack substation} to Eqn. \eqref{eq:constr power balance}}\label{eq:MINLP constraint 7}
\end{align}
\end{subequations}
where $Z$ is a sufficiently large positive constant. Optimization problem \eqref{eq:MINLP} slightly abuses the notation, and uses $\mu(z^o)$ and $\tau(z^a)$ to represent the probabilities the grid operator applies $z^o$ and the adversary applies $z^a$, respectively. Constraints \eqref{eq:MINLP constraint 1} and \eqref{eq:MINLP constraint 2} ensures that $\mu$ is a well-defined mixed strategy. Constraints \eqref{eq:MINLP constraint 3} and \eqref{eq:MINLP constraint 4} capture the fact that the adversary computes a pure strategy as its best response. Constraints \eqref{eq:MINLP constraint 5} and \eqref{eq:MINLP constraint 6} quantify the optimal power flow disruption $r$ that the adversary can cause. By Eqn. \eqref{eq:MINLP constraint 5}, we have that if the adversary plays its best response ($\tau(z^a)=1$), then it can achieve $r$ amount of power flow disruption. For $\tau(z^a)=0$, constraint \eqref{eq:MINLP constraint 5} is satisfied trivially. Constraint \eqref{eq:MINLP constraint 7} guarantees that the strategies are proper.

Optimization problem \eqref{eq:MINLP} is a mixed integer nonlinear program (MINLP). The nonlinearity can be mitigated by defining a new variable $u_{z^oz^a}$, which is defined as $u_{z^oz^a}=\mu(z^o)\tau(z^a)$ for all $z^o,z^a$ satisfying Eqn. \eqref{eq:MINLP constraint 7} and $u_{z^oz^a}=0$ otherwise. The constraints defined on $u_{z^oz^a}$ are
\begin{subequations}\label{eq:new constr}
\begin{align}
    &u_{z^oz^a}\in[0,1],~\forall z^o\in\mathcal{Z}^o,\forall z^a\in\mathcal{Z}^a,~0\leq\sum_{z^a\in\mathcal{Z}^a}u_{z^oz^a}\leq 1,~\forall z^o\in\mathcal{Z}^o\\
    &\sum_{z^o\in\mathcal{Z}^o}\sum_{z^a\in\mathcal{Z
    }^a}u_{z^oz^a}=1.
\end{align}
\end{subequations}
Using $u_{z^oz^a}$, MINLP \eqref{eq:MINLP} is converted to a mixed integer linear program (MILP):
\begin{subequations}\label{eq:MILP}
\begin{align}
    \underset{u,\tau,r}{\min}&\sum_{z^o\in\mathcal{Z}^o}\sum_{z^a\in\mathcal{Z}^a}u_{z^oz^a}\sum_{(i,j)\in\mathcal{L}}\frac{1-z^o_{i,j}z^a_{i,j}}{2}(|P_{i,j}|+|P_{j,i}|)\label{eq:MILP obj}\\
    \text{subject to}\quad& \text{Eqn. \eqref{eq:MINLP constraint 3}, \eqref{eq:MINLP constraint 4}, \eqref{eq:MINLP constraint 6}, and \eqref{eq:new constr}}\\
    &0\leq r-\sum_{z^o\in\mathcal{Z}^o}\left[\sum_{(i,j)\in\mathcal{L}}\frac{1-z^o_{i,j}z^a_{i,j}}{2}(|P_{i,j}|+|P_{j,i}|)\right]\left[\sum_{\bar{z}^a\in\mathcal{Z}^a}u_{z^o\bar{z}^a}\right]\nonumber\\
    &\quad\quad\quad\quad\quad\quad\quad\leq (1-\mathbb{P}(z^a))Z,~\forall z^a\in\mathcal{Z}^a\label{eq:MILP constr 2}
\end{align}
\end{subequations}
Similar techniques for converting MINLP to MILP have been used in \cite{paruchuri2008playing}. The equivalence between MILP \eqref{eq:MILP} and MINLP \eqref{eq:MINLP} is presented as follows.
\begin{lemma}
The MINLP \eqref{eq:MINLP} is equivalent to the MILP \eqref{eq:MILP}.
\end{lemma}
\begin{proof}
We first prove that the objective functions of MINLP \eqref{eq:MINLP} and MILP \eqref{eq:MILP} are identical. We then show that a feasible solution to Eqn. \eqref{eq:MINLP} is also feasible to Eqn. \eqref{eq:MILP}, and vice versa. The equivalence between \eqref{eq:MINLP obj} and \eqref{eq:MILP obj} holds by the construction of $u_{z^oz^a}$.

Let $\mu$, $\tau$, and $r$ be feasible solutions to Eqn. \eqref{eq:MINLP}. Let $u_{z^oz^a}=\mu(z^o)\tau(z^a)$. We have that $u_{z^oz^a}\in[0,1]$ holds by the construction of $u_{z^oz^a}$. By the definition of $u_{z^oz^a}$, we have that $\sum_{z^o\in\mathcal{Z}^o}\sum_{z^a\in\mathcal{Z}^a}u_{z^oz^a}=1$ holds by constraints \eqref{eq:MINLP constraint 1} and \eqref{eq:MINLP constraint 4}. Inequality $0\leq\sum_{z^a\in\mathcal{Z}^a}u_{z^oz^a}\leq 1$ holds by Eqn. \eqref{eq:MINLP constraint 3} and the definition of $u_{z^oz^a}$. Constraint \eqref{eq:MILP constr 2} follows by substituting $u_{z^oz^a}$ into Eqn. \eqref{eq:MINLP constraint 5}.

Let $u,\tau$, and $r$ be feasible to Eqn. \eqref{eq:MILP}. We prove that $\mu,\tau$, and $r$ are feasible solutions to Eqn. \eqref{eq:MINLP}, where $\mu(z^o)=\sum_{z^a}u_{z^oz^a}$. Since $\mu(z^o)=\sum_{z^a}u_{z^oz^a}$ and $\tau(z^a)\in\{0,1\}$, we have that constraint $\sum_{z^o\in\mathcal{Z}^o}\sum_{z^a\in\mathcal{Z}^a}u_{z^oz^a}=1$ implies that constraint \eqref{eq:MINLP constraint 1} holds. Using $\mu(z^o)=\sum_{z^a}u_{z^oz^a}$, constraint $0\leq\sum_{z^a\in\mathcal{Z}^a}u_{z^oz^a}\leq 1$ can be rewritten as $0\leq \mu(z^o)\leq 1$, i.e., constraint \eqref{eq:MINLP constraint 2}. Similarly, the equivalence between constraints \eqref{eq:MILP constr 2} and \eqref{eq:MINLP constraint 5} follows by $\mu(z^o)=\sum_{z^a\in\mathcal{Z}^a}u_{z^oz^a}$.
\end{proof}

Consider line 5 of Algorithm \ref{algo:DO}. This corresponds to Eqn. \eqref{eq:bilevel 1} to \eqref{eq:bilevel 2} when the strategy of the adversary is fixed. Given any feasible $(y,z^a)$ for the adversary, the grid operator solves the following optimization problem:
\begin{subequations}\label{eq:operator BR}
\begin{align}
    \min_{w,x,z^o,f,\tilde{P},g,\theta}&\sum_{(i,j)\in\mathcal{L}}(1-z_{i,j}^oz_{i,j}^a)\frac{|P_{i,j}|+|P_{j,i}|}{2}\label{eq:operator BR obj}\\
    \text{subject to}~&\text{Eqn. \eqref{eq:constr indicator} to Eqn. \eqref{eq:constr power balance}}
\end{align}
\end{subequations}
Eqn. \eqref{eq:operator BR} is an MILP and can be solved using commercial solvers. Note that optimization problem \eqref{eq:operator BR} computes a pure strategy $(z^o,g)$ for the grid operator.

In the following, we present an MILP for line 7 of Algorithm \ref{algo:DO}, which corresponds to solving Eqn. \eqref{eq:bilevel 3} to \eqref{eq:bilevel 4} when the strategy of the grid operator is given. Since the grid operator plays a mixed strategy, the goal of the adversary then becomes maximizing the expected power flow disruption, where the expectation is taken over mixed strategy $\mu$. With a slight abuse of notation, we denote the probability that the grid operator trips the transmission lines corresponding to $z^o$ as $\mu(z^o)$. The MILP corresponding to line 7 of Algorithm \ref{algo:DO} is given as
\begin{subequations}\label{eq:adv expected BR}
\begin{align}
    \max_{y,z^a}&\sum_{z^o\in\mathcal{Z}^o}\mu(z^o)\sum_{(i,j)\in\mathcal{L}}(1-z_{i,j}^oz_{i,j}^a)\frac{|P_{i,j}|+|P_{j,i}|}{2}\label{eq:adv expected BR obj}\\
    \text{subject to }~&\text{Eqn. \eqref{eq:constr attack substation} to Eqn. \eqref{eq:constr attack indicator}}
\end{align}
\end{subequations}

In the following, we show that the optimization problem \eqref{eq:adv expected BR} can be mapped to a submodular maximization problem subject to a cardinality constraint. As a consequence, a greedy algorithm is presented to solve for a pure strategy for the adversary. We relax optimization problem \eqref{eq:adv expected BR} as
\begin{subequations}\label{eq:adv submodular}
\begin{align}
    \max_{\hat{\mathcal{B}}}~&\sum_{z^o\in\mathcal{Z}^o}\mu(z^o)\left[\sum_{(i,j)\in\hat{\mathcal{L}}}\frac{|P_{i,j}|+|P_{j,i}|}{2}+\sum_{(i,j)\in\tilde{\mathcal{L}}}\frac{|P_{i,j}|+|P_{j,i}|}{2}\right]\label{eq:adv submodular obj}\\
    \text{subject to }~&\hat{\mathcal{L}}=\{(i,j)\in\mathcal{L}:i\in\hat{\mathcal{B}} \text{ or } j\in\hat{\mathcal{B}}\}\label{eq:adv submodular constr 1}\\
    &|\hat{\mathcal{B}}|\leq C\label{eq:adv submodular constr 2}
\end{align}
\end{subequations}
where $\tilde{\mathcal{L}}\subset\mathcal{L}$ is the set of transmission lines tripped by the grid operator when taking action $z^o$. We characterize the relation between optimization problem \eqref{eq:adv expected BR} and \eqref{eq:adv submodular} using the following lemma.
\begin{lemma}\label{lemma:worst case}
Given the strategy of the grid operator, the optimal solution to optimization problem \eqref{eq:adv submodular} is identical to that of optimization problem \eqref{eq:adv expected BR}.
\end{lemma}
\begin{proof}
We omit the proof due to space constraint.
\end{proof}

We now map optimization problem \eqref{eq:adv submodular} to a problem of maximizing a submodular function subject to a cardinality constraint. We define $\chi_{i,j}(\hat{\mathcal{B}})$ as 
\begin{equation}\label{eq:chi}
    \chi_{i,j}(\hat{\mathcal{B}})=\begin{cases}
    1&\mbox{ if }i\in\hat{\mathcal{B}} \text{ or }j\in\hat{\mathcal{B}}\\
    0&\mbox{ otherwise}
    \end{cases}
\end{equation}
Using the definition of $\chi_{i,j}(\hat{\mathcal{B}})$, optimization problem \eqref{eq:adv submodular} can be rewritten as
\begin{subequations}\label{eq:adv submodular 1}
\begin{align}
    \max_{\hat{\mathcal{B}}}~&\sum_{z^o\in\mathcal{Z}^o}\mu(z^o)\sum_{(i,j)\in\mathcal{L}}\chi_{i,j}(\hat{\mathcal{B}})\frac{|P_{i,j}|+|P_{j,i}|}{2}\label{eq:adv submodular obj 1}\\
    \text{subject to }~&|\hat{\mathcal{B}}|\leq C\label{eq:adv submodular constr 2 1}
\end{align}
\end{subequations}
We have the following result.
\begin{proposition}\label{prop:submodular}
Objective function \eqref{eq:adv submodular obj 1} is submodular and nondecreasing with respect to $\hat{\mathcal{B}}$.
\end{proposition}
\begin{proof}
We first prove Eqn. \eqref{eq:adv submodular obj 1} is submodular with respect to $\hat{\mathcal{B}}$ using the definition of submodularity. Let $\hat{\mathcal{B}}_2\subseteq\hat{\mathcal{B}}_1\subset\mathcal{B}$. By Eqn. \eqref{eq:chi}, we have that
\begin{equation}
    \chi_{i,j}(\hat{\mathcal{B}}\cup\{h\})-\chi_{i,j}(\hat{\mathcal{B}})=\begin{cases}
    1&\mbox{ if } h=i \text{ or } h=j \text{ and } i,j\notin \hat{\mathcal{B}}\\
    0&\mbox{ otherwise}
    \end{cases},~\forall\hat{\mathcal{B}}\subset\mathcal{B}
\end{equation}
Suppose that $\chi_{i,j}(\hat{\mathcal{B}}_1\cup\{h\})-\chi_{i,j}(\hat{\mathcal{B}}_1)=1$. Since $\hat{\mathcal{B}}_2\subseteq\hat{\mathcal{B}}_1$, we have that $h=i$ or $h=j$ and $i,j\notin \hat{\mathcal{B}}_2$. Then we have that $\chi_{i,j}(\hat{\mathcal{B}}_2\cup\{h\})-\chi_{i,j}(\hat{\mathcal{B}}_2)=1$. Therefore, we have that
$
    \chi_{i,j}(\hat{\mathcal{B}}_2\cup\{h\})-\chi_{i,j}(\hat{\mathcal{B}}_2)\geq \chi_{i,j}(\hat{\mathcal{B}}_1\cup\{h\})-\chi_{i,j}(\hat{\mathcal{B}}_1)
$ holds for all $\hat{\mathcal{B}}_2\subseteq \hat{\mathcal{B}}_1\subset\hat{\mathcal{B}}$,
which implies that Eqn. \eqref{eq:chi} is submodular with respect to $\hat{\mathcal{B}}$.

Consider $\hat{\mathcal{B}}_2\subset\hat{\mathcal{B}}_1$. Then there must exist some $i\in\hat{\mathcal{B}}_1$ while $i\notin\hat{\mathcal{B}}_2$. Let $j\in\mathcal{B}$ be a substation satisfying $j\notin\hat{\mathcal{B}}_1$. Using Eqn. \eqref{eq:chi}, we have that $\chi_{i,j}(\hat{\mathcal{B}}_2)=0\leq \chi_{i,j}(\hat{\mathcal{B}}_1)=1$. Let $j\in\mathcal{B}$ be a substation satisfying $j\in\hat{\mathcal{B}}_2$. Then we have  that $\chi_{i,j}(\hat{\mathcal{B}}_2)= \chi_{i,j}(\hat{\mathcal{B}}_1)=0$. If $j\in\hat{\mathcal{B}}_1$ holds while $j\notin\hat{\mathcal{B}}_2$ does not hold. Then we have that $\chi_{i,j}(\hat{\mathcal{B}}_2)=0\leq \chi_{i,j}(\hat{\mathcal{B}}_1)=1$. Summarizing the arguments above, we have that Eqn. \eqref{eq:chi} is nondecreasing with respect to $\hat{\mathcal{B}}$.

Combining the arguments above, we have that Eqn. \eqref{eq:adv submodular obj 1} is a summation of non-negative submodular and nondecreasing functions. Therefore, Eqn. \eqref{eq:adv submodular obj 1} is a submodular and nondecreasing function with respect to $\hat{\mathcal{B}}$.
\end{proof}

According to Proposition \ref{prop:submodular}, optimization problem \eqref{eq:adv expected BR} is equivalent to a submodular maximization problem with cardinality constraint. Optimization problem \eqref{eq:adv submodular 1} can be solved using a greedy algorithm in polynomial time \cite{krause2014submodular}. It has been shown that the greedy algorithm achieves $1-\frac{1}{e}$ optimality guarantee \cite{krause2014submodular}.
We conclude this section by giving the convergence and optimality of double oracle algorithm \cite{mcmahan2003planning}. We state the result in the following lemma.
\begin{lemma}
Algorithm \ref{algo:DO} converges to the Stackelberg equilibrium within finitely many iterations if the best responses in line 5 and line 7 are calculated exactly.
\end{lemma}

\section{Numerical Evaluations}\label{sec:simulation}

This section presents our simulation setup and numerical results. We use IEEE 9-bus, 14-bus, 30-bus, 39-bus, 57-bus, and 118-bus power systems in our evaluations \cite{IEEEBusRef}. All the experiments are implemented using MATLAB R2020a on a workstation with Intel(R) Xeon(R) W-2145 CPU with 3.70GHz processor and 128GB memory. Simulation codes can be found at \cite{sdinuka}.

\begin{table}[]
    \centering
\begin{tabular}{ |p{1.2cm}||p{2cm}|p{4.2cm}|p{2.2cm}|p{1.6cm}| }
 \hline
 IEEE Dataset  & Reference Generators & Coherent Generator Groups (Using Bus Indices) & Maximum \#~of Iterations & Maximum Run time\\
 \hline
 9-Bus  & 1;~3 & \{1,2\}, \{3\} & 1 & 0.07~s\\
 14-Bus  & 1;~6 & \{1:3\}, \{6,8\}  & 7  & 0.31~s\\
 30-Bus  & 1;~13;~22 & \{1,2\}, \{13\}, \{22,23,27\} & 6 & 0.43~s \\
 39-Bus  & 30;~31;~37 & \{30\}, \{31:36\}, \{37:39\} & 9 & 1.11~s\\
 57-Bus  & 1;~6;~9 &\{1:3\}, \{6,8\}, \{9,10\} & 21 & 5.23~s\\
 118-Bus  & 10; 46; 49; 87 & \{10,12,25,26,31\}, \{46\}, \{49,54,59,61,65,66,69,80\}, \{87,89,100,103,111\} & 16 &  62.70~s\\
 \hline
\end{tabular}
    \caption{First three columns show IEEE power system case study data used in the numerical evaluations. Last two columns present the maximum number of iterations and maximum run time that Algorithm~1 takes to converge. The maximum values in the last two columns are found across a set of experiments where the adversary budget ($C$) is increased from $C = 1$ until adversarial actions cause the power system to fail.}
    \label{tab:data}
\end{table}
\subsection{Simulation Setup}

We extract the topology of power system, transmission line susceptance, load demands, generator capacities, transmission line capacities, and voltage angle bounds from the IEEE case study datasets \cite{IEEEBusRef}. The power flow in the system at the initial operating point is computed using Matpower \cite{zimmerman2010matpower}. The reference generators and generator coherent groups are chosen as in Table~1.

 We initialize Algorithm~\ref{algo:DO} using adversary \emph{pure} strategies that result in the four largest DC power flow disruptions to the system, a grid operator \emph{pure} islanding strategy in the absence of any adversary (solution to optimization problem in Eqn.~\eqref{eq:operator BR} with $z_{i,j}^a=1$ for all $(i,j)\in\mathcal{L}$) and corresponding grid operator and adversary best responses, respectively.

We evaluate the performance of our approach by comparing the DC power flow disruption resulting from Algorithm~\ref{algo:DO} and a baseline case. A grid operator in the baseline case computes an islanding strategy without considering the presence of an adversary, i.e., the grid operator solves Eqn. \eqref{eq:operator BR} with $z_{i,j}^a=1$ for all $(i,j)\in\mathcal{L}$. The adversary in baseline observes the islanding strategy computed by the grid operator, and computes its best response by solving Eqn.~\eqref{eq:adv}. Let $\mathcal{S}_{\textbf{S}}$ and $\mathcal{S}_{\textbf{B}}$ denote the transmission lines tripped in the proposed model (Algorithm~\ref{algo:DO}) and baseline case, respectively. We denote the DC power flow disruption corresponding to Algorithm~\ref{algo:DO} and baseline as $R(\mathcal{S}_{\textbf{S}})$ and $R(\mathcal{S}_{\textbf{B}})$, respectively.
\subsection{Case Study Results}

Figure~\ref{fig:IEEE-39} illustrates the grid operator islanding strategy and adversary strategy obtained using Algorithm~1 on IEEE 39-bus data for adversary budget, $C=8$. The grid operator performs islanding strategy~1 with probability (w.p.) 0.27 and islanding strategy~2 w.p. 0.73. In this case study we obtain $R(\mathcal{S}_{\textbf{B}}) \approx 10.04$~GW and $R(\mathcal{S}_{\textbf{S}})\approx 10.22$~GW. Hence, by committing to the islanding strategy given by Algorithm~1, the grid operator incurs $\sim180$~MW less DC power flow disruption.

Figure~\ref{fig:performance-results}-(a) plots the reduced DC power flow disruption achieved by the grid operator via committing to an islanding strategy given by Algorithm~1 (i.e., $R(\mathcal{S}_{\textbf{B}})$ - $R(\mathcal{S}_{\textbf{S}})$) for different values of adversary budget, $C$, under each test case given in Table~1. We construct a set of \emph{attack scenarios} by increasing the values of $C$ from $C = 1$ until the value of $C$ breaks down the grid (i.e., all the generators are isolated into individual islands). The results show that the grid operator achieves a better performance by committing to a strategy of Algorithm~1 under some attack scenarios and in other scenarios the grid operator achieves the same performance as committing to a baseline strategy. Algorithm~1 and the baseline achieve same performance when equilibrium strategies of the adversary and the grid operator do not contain any common set of transmission lines.

\begin{figure}[t]
\centering
\includegraphics[width=0.9\textwidth]{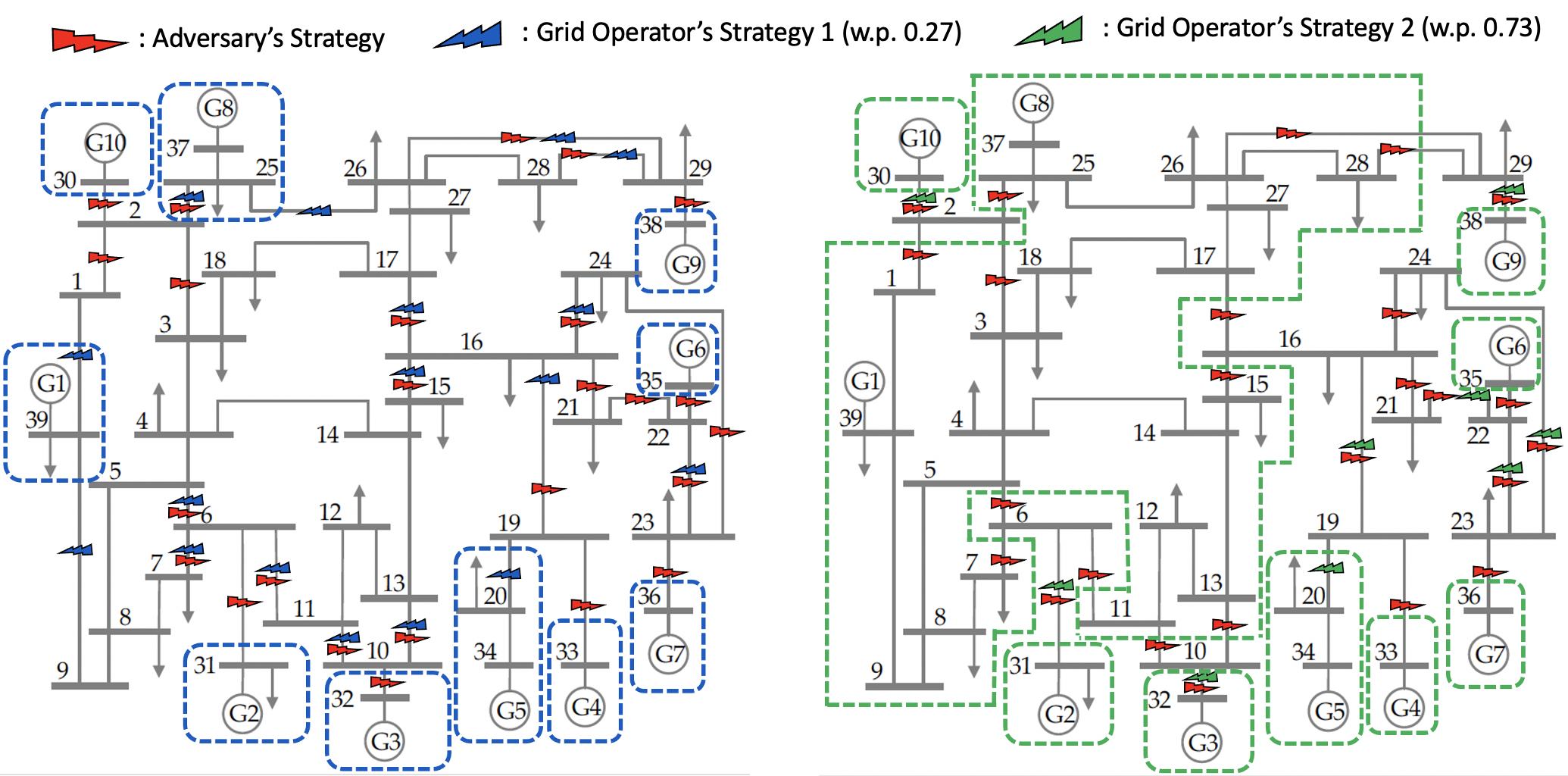}
\caption{Grid operator and adversary strategies obtained using Algorithm~1 on IEEE 39-bus data for adversary budget, $C=8$. The islands induced by the grid operator and adversary strategies are marked by the dotted lines.}
\label{fig:IEEE-39}
\end{figure}

\begin{figure}[ht]
	\centering
	\begin{subfigure}{.5\textwidth}
		\centering
		\includegraphics[width=1\textwidth]{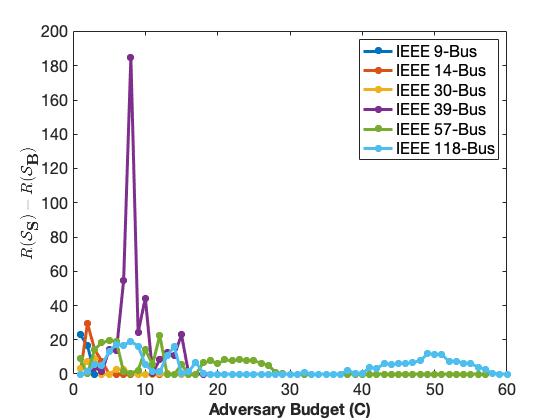}
		\caption{}
		\label{fig:budget-payoff}
	\end{subfigure}%
	\begin{subfigure}{.5\textwidth}
		\centering
		\includegraphics[width=1\textwidth]{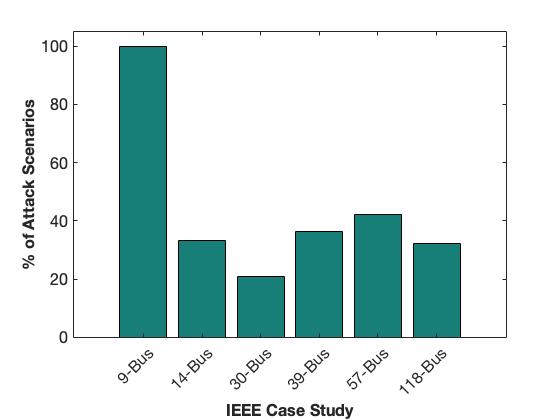}
		\caption{}
		\label{fig:percent-attk}
	\end{subfigure}
	\caption{Figure~\ref{fig:performance-results}-(a) shows the reduced DC power flow disruption achieved by the grid operator via committing to an islanding strategy given by Algorithm~1 ($R(\mathcal{S}_{\textbf{B}})$) compared to a baseline case ($R(\mathcal{S}_{\textbf{S}})$). 
	Figure~\ref{fig:performance-results}-(b) shows the percentage of attack scenarios where the grid operator performs better by committing to a strategy given by Algorithm~1.}
	\label{fig:performance-results}
\end{figure}


Figure~\ref{fig:performance-results}-(b) shows the percentage of attack scenarios where the grid operator is able to achieve lower DC power flow disruption by committing to a strategy given by Algorithm~1. We only consider the attack scenarios that does not break down the grid when computing the related percentage values. The results suggest that on average grid operator is able to perform better in $44\%$ of the attack scenarios and save $12.27$ MW when committing to a strategy of Algorithm~1.

Last two columns of Table~1 present the maximum number of iterations and maximum run time of Algorithm~1 to converge across the attack scenarios considered under each case study. 
The results show that Algorithm~1 takes less than $21$ iterations to converge for the cases analyzed. Also, run time to converge is less than $63$ seconds in Algorithm~1 for the largest dataset (IEEE-118 bus) analyzed. For other cases, Algorithm~1 finds the optimal strategies in less than $5.23$ seconds. Note that the worst-case number of iterations Algorithm \ref{algo:DO} can take to converge is $(2^{L}-1)$, (i.e., worst-case computation time is exponential in $L$). Therefore, the results suggest that Algorithm~1 converges with substantially less number of iterations compared with the worst-case bound.





\section{Conclusion}\label{sec:conclusion}

In this paper, we studied the problem of controlled islanding of a power system in the presence of a malicious adversary. We formulated the interaction between the grid operator and adversary as a Stackelberg game. The grid operator first synthesizes a mixed strategy for controlled islanding, as well as the power generation for the post-islanding system. The adversary observes the islanding strategy of the grid operator. The adversary then compromises a subset of substations in the power system and trips the transmission lines that are connected with the compromised substations. We formulated an MINLP to compute the Stackelberg equilibrium of the game. To mitigate the computational challenge incurred by solving MINLP, we proposed a double oracle algorithm based approach to solve for the equilibrium strategies. The proposed approach solved a sequence of MILPs that model the best responses for both players. Additionally, we proved that the adversary's best response can be formulated as a submodular maximization probelm under a cardinality constraint. We compared the proposed approach with a baseline, where the grid operator computes an islanding strategy by minimizing the power flow disruption without taking into account the adversary's response, using IEEE 9-bus, 14-bus, 30-bus, 39-bus, 57-bus, and 118-bus systems. The proposed approach outperformed the baseline in about $44\%$ of test cases and saved about 12.27 MW power flow disruption on average.

\bibliographystyle{splncs04}
\bibliography{GameSec21}

\end{document}